\documentclass[11pt]{amsart}

\usepackage[a4paper,left=3cm,right=3cm,top=3cm,bottom=3cm]{geometry}
\usepackage{amssymb}
\usepackage{amsthm,amsmath}
\usepackage{graphicx}
\usepackage{hyperref,color}
\usepackage{enumerate}
\newtheorem{theorem}{Theorem}

\newtheorem{proposition}{Proposition}
\newtheorem{lemma}{Lemma}
\theoremstyle{remark}

\newtheorem{remark}{Remark}

\numberwithin{equation}{section}

\newcommand{\E}{\mathbb{E}}
\newcommand{\Prob}{\mathbb{P}}
\newcommand{\pk}{\infty}
\newcommand{\pkl}{0}

\newcommand{\X}{(0,\infty)}

\renewcommand{\nu}{\mathfrak{m}}

\newcommand{\s}{\mathfrak{s}}

\newcommand{\G}{\mathcal{L}}

\begin{document}

\title[Feller diffusion and bursting in gene networks]{The combined effects of Feller diffusion and transcriptional/translational bursting in simple gene networks}

\author{Mateusz Falfus}
\address{Institute of
Mathematics, University of Silesia, Bankowa 14, 40-007 Katowice, Poland} \email{mateusz.falfus@us.edu.pl}

\author{Michael~C. Mackey}

\address{Departments of Physiology, Physics \& Mathematics, McGill University, 3655 Promenade Sir
William Osler, Montreal, QC, Canada, H3G 1Y6}\email{michael.mackey@mcgill.ca}

\author{Marta Tyran-Kami\'nska}
\address{Institute of
Mathematics, Polish Academy of Sciences, Bankowa 14, 40-007 Katowice, Poland} \email{mtyran@us.edu.pl}

\begin{abstract}
We study a stochastic model of biosynthesis of proteins in generic bacterial operons.  The stochasticity arises from two different processes, namely from  `bursting' production of either mRNA and/or protein (in the transcription/translation process) and from standard diffusive fluctuations. The amount of protein follows the Feller diffusion, while the bursting introduces random jumps between trajectories of the diffusion process. The combined effect leads to a process commonly known as a diffusion process with jumps. We study existence of invariant densities and the long time behavior  of distributions of the corresponding Markov process, proving asymptotic stability in the evolution of the density.
\end{abstract}


\maketitle

\section{Introduction}\label{s:intro}

The operon concept for the regulation of bacterial genes, which was first put forward in \cite{Jacob1960operon},  has had
an astonishing and revolutionary effect on the development of understanding in molecular biology.  In the operon concept, transcription of DNA produces messenger RNA   (mRNA, denoted $M$).  Then
through the process of translation of mRNA, intermediate protein   ($I$) is produced which is capable of
controlling metabolite ($E$) levels that in turn can feed back and affect either/or transcription and/or translation. A typical example would be in the lactose operon   where the intermediate is $\beta$-galactosidase  and the metabolite is allolactose.   These
metabolites are often referred to as effectors,  and can, in the simplest case, be either stimulatory
(so called inducible)  or inhibitory (or repressible)  to the entire process.

Mathematical treatments of the operon concept appeared rapidly after the idea was embraced by
biologists.  Thus, \cite{goodwin1965} gave the first analysis of operon dynamics
which had been formulated in \cite{Goodwin1963}.  These first attempts were swiftly followed by
Griffith's analysis of a simple repressible operon \cite{Griffith68a} and an inducible operon
\cite{Griffith68b}, and these and other results were nicely summarized in \cite{tyson-othmer-1978} and \cite{selgrade79}.

For a generic operon with a maximal level of
transcription $ b $ (in concentration units), the dynamics are given by
\cite{goodwin1965,Griffith68a,Griffith68b,othmer76,selgrade79}
   \begin{align}
    \dfrac{dM}{dt} &=   b\varphi (E) -\gamma_M M,\label{eq:mrna}\\
    \dfrac{dI}{dt} &= \beta_I    M -\gamma_I I ,\label{eq:intermed}\\
    \dfrac{dE}{dt} &= \beta_E   I -\gamma_E E.\label{eq:effector}
    \end{align}
It is assumed here that the rate $\varphi$ of mRNA production is proportional to the fraction of time the operator region
is active, and that the rates of protein and metabolite production are proportional to the amount of mRNA and
intermediate protein respectively. All three of the components $(M,I,E)$ are subject to degradation,  and the
function $\varphi $ is as determined in \cite{mackeysimple}.

Identifying fast and slow variables can give considerable simplification and insight into the long term
behavior of the system. A fast variable  in a given dynamical system relaxes much more rapidly to an equilibrium
than a slow one \cite{haken83}. Differences in degradation rates in  chemical and biochemical systems lead to
the distinction that the slowest variable  is the one that has the smallest degradation rate.  Typically the degradation rate of mRNA is much greater than the corresponding degradation rates for both the
intermediate protein and the effector $(\gamma_M \gg \gamma_I,\gamma_E)$ so in this case the mRNA dynamics are
fast and we have $\gamma_M M  \simeq  b \varphi (E)$.
If $\gamma_M \gg \gamma_I \gg \gamma_E$ so that the effector is the slowest variable, then the three variable system describing the generic operon
reduces to a one dimensional system
    \begin{equation} \dfrac{dE}{dt} = -\gamma_E E +  \dfrac{b\beta_I \beta_E}{\gamma_I \gamma_M}  \varphi (E)
    \label{eq:1D-E}
    \end{equation}
for the relatively slow effector dynamics. If instead  the effector  qualifies as a fast variable so that $\gamma_M \gg \gamma_E \gg \gamma_I$
    then the intermediate protein is the slowest variable and
    \begin{equation}
    \dfrac{dI }{dt} = -\gamma_I  I + \dfrac{b\beta_I}{\gamma_M}  \varphi \left(\frac{ \beta_E I}{\gamma_E}\right).
    \label{eq:1D-I}
    \end{equation}
Defining dimensionless variables, both equations \eqref{eq:1D-I} and  \eqref{eq:1D-E}  are seen to be of the form
    \begin{equation}
   \dfrac{dx}{dt} = - \gamma x + \varphi(x)
    \label{eq:1D-general}
    \end{equation}
and this will be our starting point in the examination of the effects of noise on the dynamics.

In cellular and molecular biology, as experimental techniques have allowed investigators to probe temporal behavior at ever finer levels, even to the level of individual molecules, the question has arisen about whether the fluctuations observed in data are measurement noise or are  playing a role in the operation of the molecular regulatory process.  Experimentalists and theoreticians alike who are interested in the regulation of gene networks are increasingly focussed on trying to access the role of various types of fluctuations on the operation and fidelity of both simple and complex gene regulatory systems.  Recent reviews \cite{kaern05,raj08,swain08c} give an interesting perspective on some of the issues confronting both experimentalists and modelers.

Typically, the discussion seems to focus on whether fluctuations  can be considered as extrinsic to the system
under consideration \cite{swain08a,ochab08,ochab10}, or whether they are an intrinsic part of the fundamental
processes they are affecting (e.g. bursting, see below). The dichotomy is rarely so sharp however, but
\cite{elowitz02} has proposed an elegant experimental technique to distinguish between the two, see also
\cite{raser04}, while \cite{scott06} and \cite{swain02a}   have laid the groundwork for a theoretical
consideration of this question.  One issue that is raised persistently in considerations of the role of
fluctuations or noise in the operation of gene regulatory networks is whether or not they are ``beneficial"
\cite{blake06} or ``detrimental" \cite{fraser04} to the operation of the system under consideration.  This is,
of course, a question of definition and not one that we will be further concerned with here.

 From a modeling perspective there have been a number of studies attempting to understand the effects of noise on gene regulatory dynamics. The now classical \cite{kepler01} laid much of the ground work for subsequent studies by its treatment of a variety of noise sources and their effect on dynamics. In \cite{mcmmtkry11} the effects of either bursting or Gaussian noise on both inducible and repressible operon models were examined.  A recent monograph \cite{mackeysimple} gives extensive background information for the history of modeling of the effects of noise in gene regulation.

Here, we consider the density of the molecular distributions  in generic bacterial operons in the
presence of `bursting' (commonly known as intrinsic noise in the biological literature) as well as inherent
(extrinsic) noise using an analytical approach. Our work is motivated by the well documented production of mRNA
and/or protein in stochastic bursts in both prokaryotes and eukaryotes
\cite{blake03,cai,chubb,golding,raj,sigal,yu}, and follows other contributions by, for example,
\cite{kepler01}, \cite{friedman06}, \cite{rudnicki07} and \cite{swain08b}.

Jump Markov processes are often used in modelling stochastic gene expressions with explicit bursting in either mRNA or proteins \cite{friedman06,golding}, and have been employed as models for genetic networks \cite{Zeisler:2008}.
Biologically, the `bursting' of mRNA or protein is simply a process in which there is a production of several molecules within a very short time. In the biological context of modelling stochastic gene expression, explicit models of bursting mRNA and/or protein production have been analyzed recently, either using a discrete  \cite{swain08b} or a continuous formalism \cite{friedman06,Lei2009,mcmmtkry11} as even more experimental evidence from single-molecule visualization techniques has revealed the ubiquitous nature of this phenomenon  \cite{Elf2007,golding,Ozbudak2002,Raj2009,raj,Suter2011,Xie2008}.

We consider the situation in which there is both bursting production of molecules
 and fluctuations in the degradation rate of  molecules, which we left unsolved in \cite{mcmmtkry11}.
The amount of protein follows the Feller diffusion on  $(0,\infty)$ defined as the solution of the one dimensional
Ito stochastic differential equation
\begin{equation}
    dX_t = -\gamma X_t dt + \sigma\sqrt{X_t}dW_t,\quad X_0=x\in (0,\infty),
    \label{e:feller-sde}
\end{equation}
where $\gamma,\sigma>0$ and  $\{W_t:t\ge 0\}$ is a one-dimensional standard Wiener process (Brownian motion).
There is existence and pathwise uniqueness of the solution of \eqref{e:feller-sde} until
the exit time $\zeta=\inf\{t>0: X_t\not\in (0,\infty)\}$. Moreover, the process is absorbed at $0$ a.s. in the
sense that if $\tau_0=\inf\{t>0: X_t=0\}$ is the first hitting time of $0$ for the process $\{X_t:0\le
t<\zeta\}$ then
\[
\Pr(\tau_0=\zeta<\infty|X_0=x)=\Pr(\lim_{t\uparrow \zeta}X_t=0|X_0=x)=1 \quad \text{for all}\quad x>0.
\]
Consequently,  for $x>0$ and $t\ge \tau_0$ we can define $X_t=0$, so that for every $x>0$ the process starting at
$x$ is defined for all times and has values in $[0,\infty)$. Since the unique solution $X$ of
\eqref{e:feller-sde} starting at $X_0=0$ is $X_t=0$, $t\ge 0$, we can extend the state space to $[0,\infty)$.

The random degradation of molecules described by the Feller diffusion is interrupted at random times
\[
0<t_1<t_2<\ldots
\]
which occur at rate $\varphi(x)$ dependent on the current amount of molecules. At each $t_k$ a  random amount $e_k$  of protein molecules is produced according to a distribution with density $h$, independently of everything else.
Consequently, the model follows Feller diffusion with additive jumps, which can be defined
as a Markov process  $Z=\{Z_t\}_{t\ge 0}$  solving the following stochastic differential
equation
\begin{equation}\label{eq_jumpdiff}
Z_t=Z_0-\gamma \int_0^t Z_{s-}ds+\sigma\int_{0}^t \sqrt{Z_s} dW_s+\int_0^t\int_{0}^{\infty}\int_{0}^{\infty} z 1_{\{r\le
\varphi(Z_{s-})\}}N(ds,dz,dr)
\end{equation}
where $N(ds,dz,dr)$ is a
Poisson random measure on $(0,\infty)\times [0,\infty)^2$ with intensity $dsh(z)dzdr$, $h$ is a probability density function on $(0,\infty)$, and $\varphi$ is a Borel measurable function locally bounded on $[0,\infty)$. We  study the process
$Z$ as in~~\eqref{eq_jumpdiff} on the state space $E=[0,\infty)$.

The process $Z$ is an example of a jump-diffusion process with jumps from the boundary and we study asymptotic properties of such process with the help of stochastic semigroups. Diffusion processes on bounded domains with random jumps from the boundary have fine ergodic and spectral properties \cite{pinsky09,pinsky09sp}. Yet, a different approach is given in \cite{bansaye11}, where the authors study extinction.

The outline of this paper is as follows.  Section \ref{sec:prelim} collects together some basic material including definitions and necessary concepts.  Section \ref{sec:feller} treats some elementary properties of the Feller diffusion, while Section \ref{sec:semigroup-feller} develops the semigroup for the Feller diffusion and properties of the semi-group that we will later need.  Section \ref{sec:diffusion-jumps} gets to the heart of the matter by treating the combined diffusion and jump processes for the specific case of transcriptional and/or translational bursting.  We prove in Section \ref{ssec:long} that there is a unique invariant density of the molecular distribution and that it is asymptotically stable.  In Section \ref{ssec:limit-expon} we explicitly assume that the distribution of transcriptional/translational bursts is exponentially distributed (as often found experimentally) and derive the differential equation that the stationary density satisfies.  Section \ref{constant} gives an explicit solution for the invariant density in the special case that the bursting rate $\varphi$ is constant.
We conclude in Section \ref{sec:disc-conclusions} with a brief summary of extant problems and directions for future research.

\section{Preliminaries}\label{sec:prelim}
In this section we collect some preliminary material.
Let $(\mathcal{X},\|\cdot \|)$ be a Banach lattice, it will be either an $L^1$ space of integrable functions or a subspace of the space of bounded measurable functions with the supremum norm.
The domain of a linear operator $A$  will be denoted by $\mathcal{D}(A)$. A linear operator $A$ is said to be \emph{positive} if
$A f\ge 0$ for all $f\in \mathcal{D}(A)$ such that $f\ge 0$.
A bounded linear operator $A$ is called a  contraction if $\|A\|\le 1$.
If for some real $\lambda$ the operator
$\lambda-A:=\lambda I-A$ is one-to-one, onto, and $(\lambda -A)^{-1}$ is a bounded linear operator, then $\lambda$
is said to belong to the resolvent set $\rho(A)$ and $R(\lambda,A):=(\lambda  -A)^{-1}$ is called the resolvent
at $\lambda$ of $A$.

A family of positive (contraction) operators $\{S(t)\}_{t\ge 0}$ on $\mathcal{X}$  is called a \emph{positive (contraction) semigroup}, if it is a
$C_0$-\emph{semigroup}, i.e.,
\begin{enumerate}[\upshape (1)]
\item $S(0)=I$ (the identity operator);
\item $S(t+s)=S(t)S(s)$ for every $s,t\ge 0$;
\item for each $f\in L^1$ the mapping $t\mapsto S(t)f$ is continuous: for each $s\ge 0$
\[
\lim_{t\to s^{+}}\|S(t)f-S(s)f\|=0.
\]
\end{enumerate}
The infinitesimal \emph{generator} of
$\{S(t)\}_{t\ge0}$ is by definition the operator $A$ with domain $\mathcal{D}(A)\subset \mathcal{X}$ defined as
\[
\begin{split}
\mathcal{D}(A)&=\{f\in \mathcal{X}: \lim_{t\downarrow 0}\frac{1}{t}(S(t)f-f) \text{ exists} \},\\
Af&=\lim_{t\downarrow 0}\frac{1}{t}(S(t)f-f),\quad f\in \mathcal{D}(A).
\end{split}
\]
The generator $A$ of a $C_0$-semigroup is closed with $\mathcal{D}(A)$ dense in $\mathcal{X}$, see e.g. \cite[Theorem II.1.4]{engelnagel00b}. If $A$ is the generator of the positive contraction semigroup $\{S(t)\}_{t\ge0}$ then
$(0,\infty)\subset \rho(A)$ and we have the integral representation
\[
R(\lambda,A)f=\int_{0}^{\infty}e^{-\lambda s}S(s)f\,ds \quad \text{for}\quad f\in \mathcal{X}.
\]
The operator $\lambda R(\lambda,A)$ is a positive contraction and $R(\mu,A)f\le R(\lambda,A)f$ for $\mu>\lambda>0$, $f\in
\mathcal{X}$, $f\ge 0$.
For the semigroup theory we refer to \cite{engelnagel00b}.

Let the triple $(E,\mathcal{E},m)$ be a $\sigma$-finite measure space and let $L^1=L^1(E,\mathcal{E},m)$ with
norm denoted by $\|\cdot\|_1$. A linear operator $P$ on $L^1$ is called \emph{substochastic} (\emph{stochastic})
if $Pf\ge 0$ and $\|Pf\|_1\le \|f\|_1$ ($\|Pf\|_1= \|f\|_1$) for all $f\ge 0$, $f \in L^1$. We denote by $D$ the
set of all densities on $E$, i.e.
$$
D=\{f\in L^1: \,\, f\ge 0,\,\, \|f\|_1=1\},
$$
so that a stochastic operator transforms a density into a density.
A  semigroup $\{P(t)\}_{t\ge0}$ of linear operators on $L^1$ is called \emph{substochastic} (\emph{stochastic})
if it is a $C_0$-semigroup and for each $t>0$ the operator $P(t)$ is substochastic (stochastic). A density
$f^*$ is called \emph{invariant} or \emph{stationary} for $\{P(t)\}_{t\ge 0}$ if $f^*$ is a fixed point of each
operator $P(t)$, $P(t) f^*=f^*$ for every $t\ge 0$.

\section{Feller diffusion}\label{sec:feller}

In this section we study  the semigroups on the spaces of continuous functions related to the Feller diffusion introduced in Section~\ref{s:intro}.
Let $X=\{X_t\}_{t\ge 0}$ be the Feller diffusion as in \eqref{e:feller-sde} defined on
$[0,\infty)$ and let
 $\Prob_x$ be the law of the process $X$
starting at $X_0=x\ge 0$. We denote by $\E_x$ the expectation with respect to $\Prob_x$.

\begin{lemma}\label{l:fdc} Let  $\tau_0=\inf\{t>0: X_t=0\}$.
For all $x>0, t>0$ and all bounded Borel measurable functions $f$ on $[0,\infty)$  we have
\[
\E_x(f(X_t)1_{\{t<\tau_0\}}) =\int_{0}^\infty p_t(x,y)f(y)dy,
\]
where
\begin{equation}
p_t(x,y)=c_t\sqrt{\frac{e^{\gamma t}x}{y}}e^{-c_t(x+e^{\gamma t}y)}I_1(2c_t\sqrt{e^{\gamma t}xy}),\quad t,x,y>0,
\label{e:ptxy}
\end{equation}
$I_1$ is the modified Bessel function of the first kind
\[
I_1(x)=\sum_{k=0}^\infty\frac{1}{k!(k+1)!}\left(\frac{x}{2}\right)^{2k+1} \quad \text{and}\quad
c_t=\frac{2\gamma}{\sigma^2(e^{\gamma t}-1)}.
\]
Moreover,
\begin{equation}
\Prob_x(t<\tau_0)=\int_{0}^\infty p_t(x,y)dy=1-e^{-c_tx},\quad x> 0, t>0,
\label{e:T0}
\end{equation}
and
\begin{equation}\label{e:ext0}
\E_x(e^{-\lambda \tau_0})=\lambda\int_{0}^{\infty}e^{-\lambda t}\Prob_x(\tau_0\le
t)dt=\frac{\lambda}{\gamma}\int_{0}^{\infty}e^{-\frac{2\gamma}{\sigma^2}xs}
\frac{s^{\frac{\lambda}{\gamma}-1}}{(s+1)^{\frac{\lambda}{\gamma}+1}}ds,\quad x>0,\lambda>0.
\end{equation}
\end{lemma}
\begin{proof}
The process $X$ can be represented as (see e.g. \cite[Section XI.1]{revuzyor99})
\[
X_t=e^{-\gamma t}Y_{r(t)},\quad \text{where }
r(t)=\frac{1}{2c_t}
\]
and $Y=\{Y_t\}_{t\ge 0}$ is the unique strong solution of the equation
\[
dY_t=2\sqrt{Y_t}dW_t, \quad Y_0=x\ge 0,
\]
called the square of the Bessel process of dimension $0$.  The transition probabilities for $Y$ are given by (\cite[Corollary XI.1.4]{revuzyor99})
\[
\Pr(Y_t\in B|Y_0=x)=\exp\left\{-\frac{x}{2t}\right\}\delta_0(B)+\int_{B}q_t^0(x,y)dy,\quad x>0,
\]
for every Borel measurable set $B\subset[0,\infty)$, where $\delta_0$ is the Dirac measure at $0$ and
\[
q^0_t(x,y)=\frac{1}{2t}\sqrt{\frac{x}{y}}\exp\left\{-\frac{x+y}{2t}\right\} I_1\left(\frac{\sqrt{xy}}{t}\right),
\quad x, y>0.
\]
Thus the process $X$ satisfies
\[
\Prob_x(X_t\in B)=
\Pr(Y_{r(t)}\in e^{\gamma t}B|Y_0=x)
\]
for all Borel measurable sets $B\subset [0,\infty)$. Hence, we conclude that $p_t(x,y)=e^{\gamma
t}q^0_{r(t)}(x,ye^{\gamma t})$ for $x,y>0$. Since $\Prob_x(t<\tau_0)=\E_x(1_{\{t<\tau_0\}}) $, we also obtain
\eqref{e:T0}. The last formula follows by making the substitution $e^{\gamma t}=1+1/s$ under the
integral.
\end{proof}

\begin{remark}
Since
\[
I_1(x)\sim \sqrt{x/2},\quad \text{as } x\to 0, \quad \text{and}\quad I_1(x)\sim e^x/\sqrt{2\pi x}, \quad
\text{as } x\to \infty,
\]
we obtain
\[
p_t(0,y):=\lim_{x\to 0}p_t(x,y)=0\quad \text{and} \quad p_t(\infty,y):=\lim_{x\to \infty}p_t(x,y)=0\quad
\text{for all } t,y>0.\] Observe also that
\begin{equation}
\int_{0}^\infty p_t(x,y)e^{-r  y}dy=\exp\left \{-\frac{r  x e^{-\gamma t}c_t}{c_t+r }\right \}, \quad
x,t,r >0,
\label{eq:fdexp}
\end{equation}
which is a direct consequence of the following formula (see \cite[Section IX.1, p. 441]{revuzyor99})
\[
\int_{0}^\infty q_t^0(x,y)e^{-r  y}dy=\exp \left \{-\frac{r  x }{1+2r  t} \right \}.
\]
\end{remark}

We identify $C[0,\infty]$ with the space of bounded continuous functions on $(0,\infty)$ for which $f(0):=\lim_{x\to
0^+}f(x)$ and $f(\infty):=\lim_{x\to \infty}f(x)$ exist and are finite. Let $C_0(0,\infty]$ denote the subspace of $C[0,\infty]$ consisting of
functions vanishing at $0$. We consider both spaces with the supremum norm.

\begin{lemma}\label{l:T}
For $f\in C[0,\infty]$ and $t> 0$ define
\begin{equation}\label{d:semT}
T(t)f(x)=\int_{0}^\infty p_t(x,y)f(y)dy,\quad x\ge 0,
\end{equation}
where $p_t$ is as in \eqref{e:ptxy}.
Then $T(t)f\in C_0(0,\infty]$ for $f\in C[0,\infty]$ and $\{T(t)\}_{t\ge 0}$ is a positive contraction semigroup on $C_0(0,\infty]$.
\end{lemma}

\begin{proof}
Each operator $T(t)$ is a positive contraction on $C_0(0,\infty]$.
It follows from \eqref{eq:fdexp} that \begin{equation}\label{e:F2}
\lim_{t\to 0}T(t)f(x)=f(x),\quad x>0,
\end{equation}
for each function $f\in C_0(0,\infty]$ of the form $f(x)=1-e^{-rx}$, where $r>0$. Since every function from $C_0(0,\infty]$ can be approximated by a linear combinations of such exponentials, we infer that \eqref{e:F2}  holds for all $f\in C_0(0,\infty]$. This  and standard arguments (see e.g. \cite[Proposition III.2.4]{revuzyor99}) show that $\{T(t)\}_{t\ge 0}$ is a positive contraction semigroup on the Banach space  $C_0(0,\infty]$.
\end{proof}

The generator of the semigroup $\{T(t)\}_{t\ge 0}$ as in \eqref{d:semT} is  the second order differential operator
\cite{feller52,mandl68,itomckean}
\begin{equation}
\G f(x)=\frac{1}{2}\sigma^2 x f''(x)-\gamma xf'(x),\quad x>0,
\label{e:sodef}
\end{equation}
defined on
\[
\mathcal{D}(\G)=\{f\in C_0(0,\infty]\cap C^2(0,\infty): \G f\in C_0(0,\infty]\}.
\]
The point $0$ is an exit boundary point for the diffusion $X$ and the point $\infty$ is a natural boundary point. This can be verified by invoking the Feller classification of boundary points \cite{feller52}. To this end we need to recall the concept of the scale function $\s $ and the speed measure $\nu $.
They are absolutely continuous with respect to the Lebesgue measure.
If the generator is of the form
\[
\G f(x)=\beta(x)f''(x)+\alpha(x)f'(x),\quad x>0,
\]
then $\s$ and $\nu$ are defined through their derivatives  given by
\[
\s'(x)=\exp\left( -\int_1^x\frac{\alpha(z)}{\beta(z)}dz\right)\quad\text{and}\quad \nu'(x)=\frac{1}{\beta(x)\s'(x)}.
\]
Thus we have
\[
\s'(x)=e^{\frac{2\gamma}{\sigma^2}(x-1)}\quad\text{and}\quad \nu'(x)=\frac{2}{\sigma^2 x}e^{-\frac{2\gamma}{\sigma^2}(x-1)}.
\]
The point $0$ is an exit boundary if and only if
$\s'\nu \in L^1(0,1)$ and $\s \nu'\not\in L^1(0,1)$, while $\infty$ is a natural boundary point if and only if $\s'\nu \not\in L^1(1,\infty)$ and $\s \nu'\not\in L^1(1,\infty)$. Observe that we can rewrite \eqref{e:ptxy} as
\[
p_t(x,y)=p(t,x,y)\nu'(y),\quad x,y>0,
\]
where
\[
p(t,x,y)=\frac{1}{2}\sigma^2 c_t e^{-\frac{2\gamma}{\sigma^2}} \sqrt{e^{\gamma t}xy}e^{-c_t(x+y)}I_1(2c_t\sqrt{e^{\gamma t}xy}),\quad t,x,y>0.
\]

If we fix $\lambda>0$, then we may introduce the Green's function $G_\lambda$ as the Laplace transform of $p_t$
\[
G_\lambda(x,y)=\int_0^{\infty} e^{-\lambda t} p_t(x,y)dt,\quad x\ge 0,y>0.
\]
Define the operator $U_\lambda\colon  C[\pkl,\pk]\to  C[\pkl,\pk]$ by
\begin{equation}
U_\lambda f(x)=\int_{\pkl}^{\pk} G_\lambda(x,y)f(y)dy.
\label{d:Ulam}
\end{equation}
The operator $\lambda U_\lambda$ is a positive contraction on $C[\pkl,\pk]$ and for every $f\in C[\pkl,\pk]$ the function
$\psi=U_\lambda f$ is a particular solution of the equation
\begin{equation}
\lambda\psi -\G \psi=f.
\label{e:rez}
\end{equation}
We recall from \cite{feller52}  that for each $\lambda>0$ the equation
\begin{equation}
\frac12\sigma^2 x\psi''(x)-\gamma x\psi'(x)=\lambda \psi(x),\quad x\in (0,\infty),
\label{e:inv}
\end{equation}
has two strictly positive continuous solutions $\psi_+(x)$ and $\psi_{-}(x)$ for $x>0$ such that $\psi_+$ is increasing, $\psi_{-}$ is decreasing, $\psi_{+}(0)=0$, $\psi_{+}(\infty)=\infty$, $\psi_{-}(0)=1/\lambda$, $\psi_{-}(\infty)=0$. Any other continuous solution of \eqref{e:inv} is a linear combination of $\psi_{+}$ and $\psi_{-}$.
The Wronskian
\[
w_\lambda(x):=\psi_{+}'(x)\psi_{-}(x)-\psi_{+}(x)\psi_{-}'(x)
\]
satisfies $w_\lambda(x)=w_\lambda(1) \s '(x)$ with $w_\lambda(1)>0$ (note that we can set $\s '(1)=1$).
We have
\begin{equation}
G_\lambda(x,y)=\frac{1}{w_\lambda(1)}\nu'(y)\left\{
                 \begin{array}{ll}
                   \psi_{+}(x)\psi_{-}(y), & \hbox{if } x\le y,\\
                   \psi_{+}(y)\psi_{-}(x), & \hbox{if } y\le x.
                 \end{array}
               \right.
               \label{e:Green}
\end{equation}

\begin{remark}\label{r:kummer} If we divide equation \eqref{e:inv} by $\sigma^2/2$, we get
\begin{equation}\label{e:inv1}
x\psi''(x)-\theta x\psi'(x)=\frac{2\lambda}{\sigma^2} \psi(x),
\end{equation}
where $\theta=\frac{2\gamma}{\sigma^2}$.
If we make the change of variables $z=\theta x=2\gamma x/\sigma^2$ in \eqref{e:inv1}, we find that $\psi(x)=w(\theta x)$, where $w$ is the solution of the Kummer differential equation
\begin{equation}\label{e:chde0}
z w''(z)-zw'(z)-\mu w(z)=0, \quad \mu=\frac{2\lambda}{\sigma^2\theta}.
\end{equation}

The confluent hypergeometric function of the second kind
\begin{equation}\label{e:chfU}
U(\mu,0,z)=\frac{1}{\Gamma(\mu)}\int_{0}^\infty e^{-zs}s^{\mu-1}(1+s)^{-\mu-1}ds, \quad \mu,z>0,
\end{equation}
is a positive decreasing solution of \eqref{e:chde0}.
We have $U(\mu,0,0)=1/\Gamma(\mu+1)$. 
The positive increasing solution of \eqref{e:chde0},
independent of $U(\mu,0,z)$ is given by the regularized  confluent hypergeometric function
\[
{}_1\tilde{F}_1(\mu,0,z)=\mu z{}_1F_1(\mu+1,2,z),
\]
where  ${}_1F_1(a,b,z)$ is the confluent hypergeometric function of the first type, noted also $M(a,b,z)$ or $\Phi(a,b,z)$, defined by
\begin{equation}
\label{M} {}_1F_1(a,b,z)=\sum_{k=0}^{\infty}\frac{(a)_k\,z^k}{(b)_k\,k!}
\end{equation}
 and $(c)_k$ is the Pochhammer symbol defined by
$$(c)_k=c(c+1)(c+2)\ldots(c+k-1), \quad(c)_0=1. $$
Consequently, the two strictly monotonic positive solutions of \eqref{e:inv} are
\[
\psi_{+}(x)=\frac{2\lambda}{\sigma^2}x {}_1F_1\left(\frac{\lambda}{\gamma}+1, 2,\frac{2\gamma}{\sigma^2}x\right),\quad \psi_{-}(x)=\frac{\Gamma(\mu)}{\gamma}U\left(\frac{\lambda}{\gamma},0,\frac{2\gamma}{\sigma^2}x\right).
\]
\end{remark}

For a bounded Borel measurable function $f$ on $[0,\infty)$ and $t> 0$, define
\[
T_0(t)f(x)=\E_x(f(X_t)),\quad x\ge 0.
\]
Since $\Prob_0(\tau_0=0)=1$, $\Prob_x(\tau_0<\infty)=1$, and $X_t=0$ for $t\ge \tau_0$,  we obtain that
$\E_x(f(X_t)1_{\{t\ge \tau_0\}})=f(0)\Prob_x(t\ge \tau_0)$.
Thus, by Lemma~\ref{l:fdc},
\begin{equation}
T_0(t)f(x)=\int_{0}^\infty p_t(x,y)f(y)dy+f(0)e^{-c_tx},\quad x\ge 0,t>0.
\label{d:To}
\end{equation}
Note that each operator $T_0(t)$ is a positive contraction on $C[0,\infty]$. Now let $f\in C[0,\infty]$. We can write $f=f_0+f(0)$, where $f_0\in C_0(0,\infty]$. Then
\[
T_0(t)f(x)=T(t)f_0(x)+ f(0),\quad x, t>0 ,
\]
and $T_0(t)f(0)=f(0)$. Thus, $T_0(t)f(x)-f(x)=T(t)f_0(x)-f_0(x)$.
Hence, $\{T_0(t)\}_{t\ge 0}$ is a positive contraction semigroup on $C[0,\infty]$ and its generator $(\G_0,\mathcal{D}(\G_0))$ is given by $\G_0f(0)=0$ and
\begin{equation}
\G_0f(x)=\G (f-f(0))(x),\quad x>0,  f\in \mathcal{D}(\G_0)=\{f\in C[0,\infty]: f-f(0)\in \mathcal{D}(\G)\},
\label{d:L0}
\end{equation}
where $(\G,\mathcal{D}(\G))$ is the generator of $\{T(t)\}_{t\ge 0}$  as in~\eqref{e:sodef}.

\begin{lemma}\label{l:T0}
Define the family of operators $U_\lambda^0\colon
C[\pkl,\pk]\to C[\pkl,\pk]$, $\lambda>0$, by
\begin{equation}
U_\lambda^0f(x)=U_\lambda f(x)+f(\pkl)\psi_\lambda(x), \quad f\in C[\pkl,\pk],
\end{equation}
where $U_\lambda$ is as in \eqref{d:Ulam} and  $\psi_\lambda$ is a positive decreasing solution of \eqref{e:inv} satisfying $\psi_\lambda(0)=1/\lambda$.
Then $U_\lambda^0$ is the resolvent of the operator $(\G_0 ,\mathcal{D}(\G_0 ))$ as defined in \eqref{d:L0} and \eqref{e:sodef}, i.e. $U_\lambda^0=(\lambda
I-\G_0 )^{-1}$, $\lambda>0$.
\end{lemma}
\begin{proof} Let $f\in C[0,\infty]$. It follows from \eqref{d:To} and \eqref{e:ext0} that
\[
\int_0^{\infty} e^{-\lambda t}T_0(t)f(x)dt=\int_0^\infty \int_0^{\infty} e^{-\lambda t}p_t(x,y) dt f(y)dy+f(0)\frac{1}{\lambda}\mathbb{E}_{x}(e^{-\lambda \tau_0}).
\]
We have $\psi_{\lambda}=\psi_{-}$ and $\psi_\lambda(x)=\frac{1}{\lambda}\mathbb{E}_{x}(e^{-\lambda \tau_0})$, by Lemma~\ref{l:fdc} and equation~\eqref{e:chfU}. Consequently,
\[
U_\lambda^0 f(x)=\int_0^{\infty} e^{-\lambda t}T_0(t)f(x)dt
\]
for all $x\ge0$.
\end{proof}

\section{Stochastic semigroup for the Feller diffusion}\label{sec:semigroup-feller}

The diffusion process is not conservative on $(0,\infty)$.  Thus to identify the stochastic semigroup connected
with the diffusion on $[0,\infty)$ we have to consider the space $L^1(m)=L^1([0,\infty),m)$, where $m$ is the
measure on $E=[0,\infty)$ equal to the sum of the Lebesgue measure on $(0,\infty)$ and the Dirac measure
$\delta_0$ at $0$. We can identify $L^1(m)$ with the product space $L^1(0,\infty)\times \mathbb{R}$. Thus every element $g\in L^1(m)$ can be written as $g=(u,v)$ with $u\in L^1(0,\infty)$ and $v\in \mathbb{R}$, we write
$g(x)=u(x)$ for $x>0$ and $g(0)=v$,
and we have
\[
\int_{[\pkl,\pk)}g(x)m(dx)=\int_{\pkl}^{\pk}g(x) dx +g(\pkl).
\]
If $\mu$ is a probability distribution of a nonnegative random variable $\xi$ which is absolutely
continuous with respect to $m$, so that there is a nonnegative $g\in L^1(m)$ such that
\[
\mu(B)=\int_B g(x)m(dx),\quad B\in \mathcal{B}([0,\infty)),
\]
and $\mu([0,\infty))=1$, then we say that $g$ is the \emph{density of $\xi$ with respect to $m$}.

 Let $X=\{X_t\}_{t\ge 0}$ be the Feller diffusion as in \eqref{e:feller-sde} defined on
$[0,\infty)$ and $\{T_0(t)\}_{t\ge 0}$ be the semigroup defined by \eqref{d:To}.
 We have the representation
\begin{equation*}
T_0(t)f(x)=\int_{[0,\infty)}p_t^0(x,y)f(y)m(dy),\quad x\in [0,\infty),t>0, 
\end{equation*}
for bounded Borel measurable functions, where
\begin{equation}
p_t^0(x,y)=\left\{
             \begin{array}{ll}
             p_t(x,y), & x\ge 0,y>0, \\
               e^{-c_tx}, & x\ge 0, y=0.\\
             \end{array}
           \right.
           \label{d:pto}
\end{equation}
If $\mu_0$ is the distribution of $X_0\ge 0$ then the distribution of $X_t$ is given by
\[
\mu_t(B):=\Prob_{\mu_0}(X_t\in B)=\int_{[0,\infty)}T_0(t)1_B(x)\mu_0(dx), \quad t>0, B\in
\mathcal{B}([0,\infty)).
\]
Hence, if there is a nonnegative $g\in L^1([0,\infty),m)$ such that
\[
\mu_0(B)=\int_B g(x)m(dx),
\]
then
\[
\begin{split}
\mu_t(B)&=\int_{[0,\infty)}\int_Bp_t^0(x,y)m(dy)g(x)m(dx)
= \int_{B}\int_{[0,\infty)}p_t^0(x,y)g(x)m(dx)m(dy)
\end{split}
\]
which shows that $\mu_t$  is absolutely continuous with respect to $m$ and
the density of  $X_t$ is given by
\begin{equation}
P_0(t)g(y)=\int_{[0,\infty)}p_t^0(x,y)g(x)m(dx),\quad y\ge 0.
\label{e:dbP0}
\end{equation}

We now show that $\{P_0(t)\}_{t\ge 0}$ is a stochastic semigroup on $L^1(m)$ and we find its generator.
To this end let us define the second order differential operator $A$ by
\begin{equation}\label{d:oA}
A f(x)=\frac{d}{dx}\left(\frac{d}{dx}(\beta(x) f(x))-\alpha(x) f(x)\right),\quad
\beta(x)=\frac{\sigma^2}{2}x,\,\,\alpha(x)=-\gamma x,
\end{equation}
which is meaningful for any locally integrable  function $f$ for which $\beta f$ and $(\beta f)'-\alpha f$ are
absolutely continuous.
We consider $A$ on the maximal domain
\[
\mathcal{D}_{M}(A)=\{f\in L^1(0,\infty):  A f\in L^1(0,\infty)\}.
\]
Since $A f$ is integrable for $f\in\mathcal{D}_M(A)$, the limits
\begin{equation}
n_0(f):=\lim_{x\to 0^+}((\beta f)'(x)-\alpha(x)f(x))\quad \text{and}\quad n_\infty(f):=\lim_{x\to
+\infty}((\beta f)'(x)-\alpha(x)f(x))
\end{equation}
exist and are finite. Hence
\[
\int_{0}^\infty A f(x)dx=n_\infty(f)-n_0(f),\quad f\in \mathcal{D}_M(A).
\]
Since  $\s',\nu'\not\in L^1(1,\infty)$, we have $n_\infty(f)=0$ for $f\in \mathcal{D}_M(A)$, see e.g. \cite[Lemma 1.2]{attalienti04}.

\begin{theorem}\label{t:gdp} Let $\{P_0(t)\}_{t\ge 0}$ be defined as in \eqref{e:dbP0}. Then $\{P_0(t)\}_{t\ge 0}$ is a stochastic semigroup on $L^1(m)$ and  its generator
is  the operator $A_0$ defined by
\begin{equation}
A_0g(x)=A(u)(x), \quad x>0, \quad A_0g(0)=n_0(u)
\label{d:A0}
\end{equation}
for $g(x)=u(x)$ for $x>0$ with $u\in \mathcal{D}_M(A)$.
\end{theorem}
\begin{proof} It is known  (see e.g. \cite[Theorem 15.2]{feller52}, \cite[Theorem 8.5]{hille54})
 that the linear operator $A\colon \mathcal{D}_{M}(A)\to L^1(\pkl,\pk)$  is the generator of a substochastic semigroup $\{S(t)\}_{t\ge 0}$  on $L^1(\pkl,\pk)$.
The resolvent at $\lambda>0$ of the operator $(A,\mathcal{D}_M(A))$ is given by
\[
R(\lambda,A) f(y)=\int_{\pkl}^{\pk}G_\lambda(x,y)f(x)dx,
\]
where $G_\lambda$ is as in \eqref{e:Green}. Thus
\[
S(t)f(y)=\int_{0}^\infty p_t(x,y)f(x)dx,\quad f\in L^1(0,\infty).
\]
Hence, for $g$ with $g(x)=u(x)$ for $x>0$ and $u\in L^1(0,\infty)$ and $g(0)=v$ with $v\in \mathbb{R}$, we  have
\[
P_0(t)g(y)=S(t)u(y),\quad  y\in (0,\infty),
\]
and
\[
P_0(t)g(0)=\int_{0}^\infty p_t^0(x,0)u(x)dx +p_t^0(0,0)v.
\]
From \eqref{d:pto} and \eqref{e:T0} it follows that
\[
\int_{0}^\infty p_t^0(x,0)u(x)dx=\int_{0}^\infty g(x)dx -\int_{0}^\infty\int_{0}^\infty p_t(x,y)dy u(x)dx,
\]
which implies that
\[
P_0(t)g(0)=\int_{0}^\infty g(x)m(dx)-\int_{0}^\infty S(t)u(y) dy.
\]
We have
\[
|P_0(t)g(0)-g(0)|=\left|\int_{0}^\infty u(y)dy-\int_{0}^\infty S(t)u(y) dy\right|\le \|u-S(t)u\|.
\]
Consequently, $\{P_0(t)\}_{t\ge 0}$ is a stochastic semigroup  on $L^1(m)$.

Let $(\widetilde{A_0},\mathcal{D}(\widetilde{A_0}))$ be the generator of $\{P_0(t)\}_{t\ge 0}$. It remains to show that
$\widetilde{A_0}=A_0$ and $\mathcal{D}(\widetilde{A_0})=\mathcal{D}(A_0)$.
Observe that
\[
A_0g(0)=-\int_{0}^\infty A u(y)dy
\]
if $g(x)=u(x)$ for $x>0$ with $u\in \mathcal{D}_M(A)$.
Since the operator $(A,\mathcal{D}_M(A))$ is a generator, it is closed and
we see that the operator $(A_0,\mathcal{D}(A_0))$ is closed.
We have
\[
\int_{\pkl}^{\pk} \left|\frac{S(t)u(x)-u(x)}{t}-Au(x)\right|dx \to 0, \quad
\text{as }t\to 0^{+},
\]
for $u\in \mathcal{D}_M(A)$. We also obtain
\[
\frac{P_0(t)g(\pkl)-g(\pkl)}{t}=\frac{1}{t}\left(\int_{\pkl}^\pk  u(x)dx-\int_{\pkl}^{\pk}  S(t)u(x)dx
\right)\to - \int_{\pkl}^{\pk}  A u(x)dx.
\]
This shows that
\[
\left\|\frac{P_0(t)g-g}{t}-A_0(g)\right\|_1\to 0,\quad \text{as }t\to 0^{+}
\]
for all $g\in\mathcal{D}(A_0)$. Hence $ \mathcal{D}(A_0)\subseteq \mathcal{D}(\widetilde{A_0})$ and
$A_0g=\widetilde{A_0}g$ for $g\in\mathcal{D}(A_0)$. Finally, since $A$ is the generator of a substochastic semigroup, we have $1\in \rho(A)$. It is easily seen that $1\in \rho(A_0)$, implying  $\mathcal{D}(\widetilde{A_0})=\mathcal{D}(A_0)$.
\end{proof}

\begin{remark}
Observe that
\[
R(\lambda,A_0)g(x)=R(\lambda,A) u(x),\quad x\in \X,
\]
when $g(x)=u(x)$ for $x>0$, and
\[
R(\lambda,A_0)g(\pkl)=\frac{1}{\lambda}\int_{[\pkl,\pk)}g(y)m(dy)-\int_{\pkl}^{\pk} R(\lambda,A_0)g(y) dy,\quad g\in L^1(m).
\]
\end{remark}

\section{Diffusion with jumps}\label{sec:diffusion-jumps}

Let $Z=\{Z_t\}_{t\ge 0}$ be a jump-diffusion Markov process as given by the solutions of the  equation
\eqref{eq_jumpdiff}. The process $Z$ has values in $[0,\infty)$ and it is strong Markov.  For a general approach to the problem of existence and nonnegativity of solutions to equations such as \eqref{eq_jumpdiff} we refer the reader to \cite{fuli10}. We first provide a heuristic description of a construction of solutions of \eqref{eq_jumpdiff}.
Let $\{\varepsilon_k\}_{k\ge 1}$ be a sequence of
positive independent random variables with probability density function $h$, which are also independent of
$Z_0$ and of Brownian motion $\{W_t\}_{t\ge 0}$. We assume that $\varphi$ is a nonnegative, continuous and bounded function defined on $E=[0,\infty)$. We define the process
\[
\Lambda_t=\int_{0}^t\varphi(X_s)ds,
\]
where $X$ is the Feller diffusion, and we assume that it is such that for any $x\in (0,\infty)$
\begin{equation}
\Prob_x(\lim_{t\to\infty}\Lambda_t=\infty)=1,\quad x\ge 0.
\label{e:Lambdat}
\end{equation}
We can construct the process $Z$ in \eqref{eq_jumpdiff} as follows. Let $t_0=0$ and $Z_0=x$. Given the solution $X_t$  of
\eqref{e:feller-sde} with initial condition $X_0=x$, we let $t_1=t_0+\Delta t_1$, where $\Delta t_1$ is a random
variable such that
\[
\Pr(\Delta t_1> t|X_0=x)=\mathbb{E}_x(e^{-\Lambda_t}).
\]
Starting from $Z_0=x$ we define $Z_t$ to be $X_t$ for $t<t_1$ and
\begin{equation}
Z_{t_1}=X_{t_1}+\varepsilon_1.
\end{equation}
We restart the process from $X_0=Z_{t_1}$ by following the path of the diffusion up to the next jump time
$t_2=t_1+\Delta t_2$ and at the jump time $t_2$ we add $\varepsilon_2$, and so on. In this way we define a sequence of jump times $(t_n)_{n\ge 1}$ such that $Z_{t}=X_{t-t_n}$ for $t\in [t_n,t_{n+1})$, where $X_t$ is the Feller diffusion starting at $X_0=Z_{t_n}$ and
\[
Z_{t_{n+1}}=X_{t_{n+1}-t_n}+\varepsilon_{n+1},\quad n\ge 0.
\]
Observe that if $\varphi(0)>0$ then \eqref{e:Lambdat} holds, since $\mathbb{P}_x(\tau_0<\infty)=1$ and $X_t=0$ for $t>\tau_0$.

Recall that an operator  $\widetilde{L}$ is the extended generator of the $E$-valued Markov process $Z$ as in \eqref{eq_jumpdiff}, if
its domain  $\mathcal{D}(\widetilde{L})$ consists of those measurable $f\colon E\to \mathbb{R}$ for which there exists a measurable $\widetilde{f}\colon E\to \mathbb{R}$ such that for each $z\in E$, $t>0$,
\[
\mathbb{E}_z(f(Z_t))=f(z)+\mathbb{E}_z\left(\int_{0}^t \widetilde{f}(Z_s)\,ds\right)
\]
and
\[
\int_{0}^t \mathbb{E}_z(|\widetilde{f}(Z_s)|)ds<\infty,
\]
in which case we define $\widetilde{L}f=\widetilde{f}$. It follows from \eqref{eq_jumpdiff} and the generalized It\^{o} formula that
\begin{equation}
\widetilde{L}f(x)=\frac{1}{2}\sigma^2 x f''(x)-\gamma xf'(x) +\int_0^{\infty}(f(x+y)-f(x))\varphi(x)h(y)dy
\label{d:exge}
\end{equation}
for $f\in C^2[0,\infty)$ satisfying
\begin{equation}\label{d:exged}
\mathbb{E}_z\left(\sum_{t_n\le t}|f(Z_{t_n})-f(Z_{t_n^-})|\right)<\infty
\end{equation}
for all $t>0$ and $z$.

We now show that there is a stochastic semigroup $\{P(t)\}_{t\ge 0}$ on $L^1(m)$ such that for any Borel set $F\subset E$  and any density $g\in L^1(m)$ we have
\begin{equation}
\int_E \mathbb{P}_z(Z_t\in F)g(z)m(dz)=\int_F P(t)g(x)m(dx),\quad t\ge 0.
\label{e:absc}
\end{equation}
Note that if $\xi$ and $\varepsilon$ are independent random
variables,  $\xi$ has the distribution with density $g\in L^1(m)$
and $\varepsilon$ is distributed with density $h$ on $(0,\infty)$
\[
\int_{0}^\infty h(y)dy=1,
\]
then the
distribution of $\xi+\varepsilon$ has the density $P g\in L^1(m)$ of the form
\begin{equation}
P g(x)=\int_{[0,x]}h(x-y)g(y)m(dy)=g(0)h(x)+\int_{0}^x h(x-y)g(y)dy.
\label{e:oP}
\end{equation}
The probability density function $h$ of the random variable $\varepsilon$ can be formally extended to $[0,\infty)$ by setting $h(0)=0$, so that $h\in L^1(m)$ represents a density.
Note that $P$ is a stochastic operator on $L^1(m)$, since
\[
\int_{[0,\infty)}P g(x)m(dx)=g(0)+\int_{0}^\infty\int_{0}^x h(x-y)g(y)dydx=g(0)+\int_{0}^\infty g(y)dy.
\]

\begin{theorem}\label{t:Pt}
Assume that  $\varphi\in C[0,\infty]$ and $\varphi(0)>0$. Suppose that $Z_0$ has a density $g$ with respect to $m$. Then the distribution of
$Z_t$ has a density $P(t)g$ with respect to $m$, where $\{P(t)\}_{t\ge 0}$ is a stochastic semigroup on $L^1(m)$ with
generator
\begin{equation}
G g=A_0 g-\varphi g+P(\varphi g), \quad  g\in \mathcal{D}(A_0),
\label{e:genP}
\end{equation}
the operator $(A_0,\mathcal{D}(A_0))$ is  as in Theorem~\ref{t:gdp} and $P$ is the stochastic operator as
in~\eqref{e:oP}.
\end{theorem}
\begin{proof}
Since $\varphi\in C[0,\infty]$, $\varphi$ is bounded and we can write
\begin{equation*}
G g=A_0g-\lambda g+\lambda P_\lambda g, \quad g\in \mathcal{D}(A_0),
\end{equation*}
where $\lambda>0$ is any positive constant such that $\lambda>\sup_{x}\varphi(x)$, and $P_\lambda$ is a stochastic
operator of the form
\[
P_\lambda g=\left(1-\frac{\varphi}{\lambda}\right)g+P\left(\frac{\varphi}{\lambda} g\right).
\]
From the Phillips perturbation theorem \cite[Theorem 7.9.1]{almcmbk94}  it follows that $(G,\mathcal{D}(A_0))$ is the generator of the stochastic
semigroup $\{P(t)\}_{t\ge 0}$ given by
\begin{equation}
P(t)g= e^{-\lambda t}\sum_{n=0}^{\infty}\lambda^n S_n(t)g,
\label{Phil}
\end{equation}
where $S_0(t)=P_0(t)$ with $P_0(t)$ defined as in \eqref{e:dbP0} and
\[
S_{n+1}(t)g= \int_0^t S_0(t-s)P_\lambda S_n(s)g \,ds, \quad n\ge0.
\]
It follows from Lemma~\ref{l:T0} and Theorem~\ref{t:gdp} that
\[
\int_{[0,\infty)} f(x)A_0g(x)m(dx)=\int_{[0,\infty)} \G_0 f(x)g(x)m(dx)
\]
for all $f\in \mathcal{D}(\G_0)$ and $g\in \mathcal{D}(A_0)$, which implies that
\[
\int_{[0,\infty)}  f(x)G g(x)m(dx)=\int_{[0,\infty)} (\G_0 f(x)-\lambda f(x)+\lambda P_\lambda^*f(x)) g(x)m(dx),
\]
where
\[
P_\lambda^*f(x)=\left(1-\frac{\varphi(x)}{\lambda}\right)f(x)+\frac{\varphi(x)}{\lambda}\int_0^\infty f(x+y)h(y)dy
\]
for all bounded and measurable functions. Since $P_\lambda^*(C[0,\infty])\subseteq C[0,\infty]$, we conclude that the operator $Lf=\G_0 f-\lambda f+\lambda P_\lambda^*f$ is the generator of a positive contraction semigroup on $C[0,\infty]$.
Observe that for $f\in \mathcal{D}(\mathcal{L}_0)$ we have $Lf(x)=\widetilde{L}(f)(x)$ for $x\ge 0$, where $\widetilde{L}$ is the extended generator as in \eqref{d:exge}. This implies that
\begin{equation*}
Lf(x)=\lim_{t\to 0}\frac{\mathbb{E}_xf(Z_t)-f(x)}{t},\quad x\ge 0,
\end{equation*}
for $f\in \mathcal{D}(\mathcal{L}_0)$.
Consequently, we obtain
\[
\int_{[0,\infty)}f(x)P(t)g(x)m(dx)=\int_{[0,\infty)} \mathbb{E}_x(f(Z_t))g(x)m(dx)
\]
for all $f\in C[0,\infty]$ and $g\in L^1(m)$. Since an indicator function of a closed set can be approximated by globally Lipschitz continuous functions, equality \eqref{e:absc} holds for all closed sets, implying that \eqref{e:absc} holds for all Borel subsets of $[0,\infty)$ and completing the proof.
\end{proof}

\subsection{Long term behavior of the solutions}\label{ssec:long}

We now study the long term behavior of the semigroup $\{P(t)\}_{t\ge 0}$ with generator $G$ as in \eqref{e:genP}.   Let  $u(t,x)=P(t)g(x)$ for $x>0$ and $v(t)=P(t)g(0)$. Then
\begin{equation}
\begin{split}
 \dfrac{\partial u (t,x)}{\partial t} &=\frac{\sigma^2}{2}\frac{\partial^2 (xu(t,x))}{\partial x^2}+\gamma\dfrac{\partial (x u(t,x))}{\partial x} -\varphi(x) u(t,x)\\ &\quad  +  h(x)\varphi(0)v(t)+\int_0^{x} h(x-y) \varphi(y)u(t,y)dy,\\
 \dfrac{d v(t)}{d t}&=\lim_{x\to 0^+}\left(\frac{\sigma^2}{2}\frac{\partial (x u(t,x))}{\partial x}+\gamma x u(t,x)\right)-\varphi(0)v(t).
\end{split}
\label{eq:judi}
\end{equation}

\begin{theorem}
If the function $\varphi\in C[0,\infty]$ is such that $\varphi(0)>0$ and if $\int_0^\infty x h(x)dx<\infty$, then there exists a unique invariant density $g_*$ for the semigroup $\{P(t)\}_{t\ge 0}$, $g_*$ is strictly positive and the semigroup $\{P(t)\}_{t\ge 0}$ is asymptotically stable, i.e.,
\[
\lim_{t\to \infty}\int_{[0,\infty)} |P(t)g(x)-g_*(x)|m(dx)=0
\]
for all densities $g\in L^1(m)$.
\end{theorem}
\begin{proof}
It follows from \eqref{Phil} that
\[
P(t)g(x)\ge \int_{[0,\infty)}k(t,x,y)g(y)m(dy)
\]
for all nonnegative $g\in L^1(m)$, where
\[
k(t,x,y)=e^{-\lambda t}p_t^0(y,x),\quad x,y,t\ge 0.
\]
We first show that either $\{P(t)\}_{t\ge0}$ is asymptotically stable or it is sweeping from compact subsets of $[0,\infty)$, i.e.,
\begin{equation}
\lim_{t\to\infty}\int_FP(t)g(z)m(dz)=0
\label{e:sweep}
\end{equation}
for all compact sets $F\subset E$ and all densities $g\in L^1(m)$.
To this end we need to check, by \cite[Corollary 5.4]{rudnickityran17}, that the semigroup $\{P(t)\}_{t\ge 0}$
satisfies condition
\begin{enumerate}
\item[(K)] for every $y_0\ge 0$ we can find $\varepsilon >0$, $t>0$ and a measurable nonnegative function $\eta$ defined on $[0,\infty)$ such that $\int_{[0,\infty)} \eta(x)m(dx)>0$ and
    \[
    k(t,x,y)\ge \eta(x)1_{B(y_0,\varepsilon)}(y) \quad \text{for all }x,y\ge 0,
    \]
    where $B(y_0,\varepsilon)$ is an open ball in the space $[0,\infty)$ with center $y_0$ and radius $\varepsilon$,
\end{enumerate}
and that $\{P(t)\}_{t\ge 0}$
is irreducible, i.e.,
$\int_0^\infty P(t)g(x)dt>0$ for almost all $x\in [0,\infty)$ and any density $g$.
To check condition (K) observe that for every  $t>0$ the function $(x,y)\mapsto p_t(x,y)$ is strictly positive and continuous on $(0,\infty)\times (0,\infty)$.  Thus for every $y_0>0$ and $t>0$ we can find constants $\varepsilon>0$ and $c>0$ such that $k(t,x,y)\ge c$ for all $(x,y)\in (y_0-\varepsilon,y_0+\varepsilon)^2$. For $y_0=0$ there exists $\varepsilon >0$ such that we have $k(t,x,y)\ge \eta(x)$ for $y\in [0,\varepsilon)$, where $\eta(x)=0$ for $x>0$ and $\eta(0)=e^{-c_t\varepsilon-\lambda t}$.
We now show that $\{P(t)\}_{t\ge 0}$ is irreducible.
Note that for any $t>0$ and any density $g$ we have
\[
P_0(t)g(0)=\int_{[0,\infty)}e^{-c_t x}g(x)m(dx)>0,
\]
since $e^{-c_t x}>0$ and $g\neq 0$. Thus $P(t)g(0)>0$ for all $t>0$. If $m\{x>0:g(x)\neq 0\}>0$, then $P(t)g(x)>0$ for almost all $x>0$ and all $t>0$, since $k(t,x,y)>0$ for all $t,x,y>0$.  It remains to check positivity of $P(t)g$ when $g(x)=0$ for all $x>0$. Observe that $S_0(s)g(0)=g(0)$ and  $P_\lambda g(x)\ge \varphi(0)g(0)h(x)/\lambda$. Thus,
\[
S_1(t)g(x)=\int_0^t S_0(t-s)P_\lambda (S_0(s)g)(x)ds\ge \frac{\varphi(0)g(0)}{\lambda}\int_0^t \int_{[0,\infty)}p_t^0(y,x)h(y)dy ds>0
\]
for all $t>0$ and $x>0$. This together with \eqref{Phil} implies that $P(t)g(x)>0$ for all $t>0$ and $x>0$, completing the proof of irreducibility.

Next, we show that the process is not sweeping from compact subsets of $[0,\infty)$. Suppose, contrary to our claim, that the process is sweeping. It follows from \eqref{e:sweep} that for every compact set $F$ and every density $g$ we have
\[
\lim_{t\to\infty}\frac{1}{t}\int_0^t \int_E \mathbb{P}_z(Z_s\in F)g(z)m(dz)ds=0.
\]
The Chebyshev inequality implies that
\[
\mathbb{P}_z(Z_t\in F_a)\ge 1-\frac{1}{a}\mathbb{E}_zV(Z_t)
\]
for all $t>0$, $z\in E$, and $a>0$, where $V$ is a nonnegative measurable function and $F_a=\{z\in E: V(z)\le a\}$. To get a contradiction it is enough to show that
 \[
\limsup_{t\to\infty}\frac{1}{t}\int_0^t\int_E \mathbb{E}_zV(Z_s)g(z)m(dz)ds<\infty
\]
for a density $g$ and a continuous function $V$ such that each  $F_a$ is a compact subset of $E$ for all sufficiently large $a>0$. It follows from \eqref{d:exge} and \eqref{d:exged} that for $V(x)=x$ we have
\begin{equation}
\widetilde{L}V(x)=-\gamma V(x)+\varphi(x)\int_0^\infty yh(y)dy,\quad x>0.
\end{equation}
Since the function $\varphi$ is bounded, there exist a constant $c>0$ such that $\widetilde{L}V(x)\le -\gamma V(x)+c$ for $x>0$. Hence, we obtain
\[
0\le \mathbb{E}_zV(Z_t)\le z+ \mathbb{E}_z\int_0^t (-\gamma V(Z_s)+c)ds.
\]
Consequently,
\[
\frac{1}{t}\int_0^t\mathbb{E}_z V(Z_s) ds\le \frac{z}{\gamma t}+\frac{c}{\gamma}.
\]
Now, if we take a density $g\in L^1(m)$ such that $\int_0^\infty z g(z)dz<\infty$, then the claim follows, which completes the proof.
\end{proof}

\subsection{Limiting behavior of solutions with exponentially distributed bursting}\label{ssec:limit-expon}

We now look for an equation for a stationary density $g_*$ when $h$ is given by
   \begin{equation}
    h(x) = \dfrac 1 { b} e^{-x/{ b}}1_{(0,\infty)}(x),
    \label{eq:bursting-den}
    \end{equation}
with $b>0$.  We have selected the exponential form for $h$ because of the well documented \cite{blake03,cai,chubb,golding,raj,sigal,yu} exponential distribution of burst amplitudes in experimental studies.

\begin{proposition} Suppose that  $\varphi\in C[0,\infty]$ is such that $\varphi(0)>0$.
The stationary positive integrable solution  of \eqref{eq:judi}  is given by
\begin{equation}
u(x)=e^{-x/b}\frac{y(x)}{x},\quad x>0,  \quad v=n_0(u)/\varphi(0),
\label{e:invd}
\end{equation}
where $y$ is a positive solution of the  differential equation
\begin{equation}
y''(x)-\theta y'(x)=\frac{2\varphi(x)}{\sigma^2 x}y(x)
\quad \text{with}\quad \theta=\frac{1}{b}-\frac{2\gamma}{\sigma^2}
\label{e:chek}
\end{equation}
such that
\[
n_0(u)=\frac{\sigma^2}{2}\lim_{x\to
0}\left(y'(x)-\theta y(x)\right)\neq 0.
\]
\label{prop:1}
\end{proposition}

\begin{proof}
To find $g$ with $g(x)=u(x)$ for $x>0$ and $g(0)=v$  we need to solve
the equation $G g=0$, where $G$ is the generator of the semigroup $\{P(t)\}_{t\ge 0}$ given by \eqref{e:genP}. We have
\begin{equation}
G g(x)=A u(x)-\varphi(x) u(x)+\int_{0}^x \varphi(y) u(y)\frac{1}{b}e^{-(x-y)/b}dy+\varphi(0)v\frac{1}{b}e^{-x/b}=0
\label{e:ssol}
\end{equation}
for $x>0$ and $G g(0)=n_0(u)-\varphi(0)v=0$, where $A$ is as in \eqref{d:oA}.
Since
\[
A u(x)=\frac{d}{dx}\frac{1}{\s'(x)}\frac{d}{dx}(\beta(x) \s'(x) u(x)),\quad \text{where}\quad
\s'(x)=e^{-\int^x \frac{\alpha(z)}{\beta(z)}dz}
\]
and $\alpha(x)=-\gamma x$, $\beta(x)=\sigma^2 x/2$,
we get
\[
\int_0^xA u(z)dz=\frac{1}{\s'(x)}\frac{d}{dx}(\beta(x) \s'(x) u(x))-n_0(u).
\]
Observe that
\[
\int_0^x\left(\int_{0}^z \varphi(y) u(y)\frac{1}{b}e^{-(z-y)/b}dy-\varphi(z) u(z)\right)dz=-\int_0^x\varphi(y)
u(y)e^{-(x-y)/b}dy.
\]
This together with \eqref{e:ssol} gives
\[
\frac{1}{\s'(x)}\frac{d}{dx}(\beta(x) \s'(x) u(x))-\int_{0}^x \varphi(y) u(y)e^{-(x-y)/b}dy-n_0(u)e^{-x/b}=0
\]
for $x>0$.
Hence we obtain
\begin{equation}
A u(x)-\varphi(x) u(x)+\frac{1}{b}\frac{1}{\s'(x)}\frac{d}{dx}(\beta(x) \s'(x) u(x))=0
\label{e:st1}
\end{equation}
and, multiplying \eqref{e:st1} by $e^{x/b}$, leads to
\begin{equation}
\frac{d}{dx}\left(\frac{e^{x/b}}{\s'(x)}\frac{d}{dx}(\beta(x) \s'(x) u(x))\right)=\varphi(x) u(x)e^{x/b}.
\label{sturm}
\end{equation}
If we take
\[
\s_1'(x)=e^{-x/b}\s'(x), \quad f(x)=e^{x/b}u(x),
\]
then equation \eqref{sturm} becomes
\[
\frac{d}{dx}\left(\frac{1}{\s_1'(x)}\frac{d}{dx}(\beta(x) \s_1'(x) f(x))\right)=\varphi(x) f(x).
\]
Let
\[
A_1 f(x)=\frac{d}{dx}\left((\beta(x)f(x))'-\alpha_1(x)f(x)\right),\quad \alpha_1(x)=\alpha(x)+\frac{1}{b}\beta(x)=(\frac{\sigma^2}{2b}-\gamma)x.
\]
Then we have
\[
A_1 f(x)=\varphi(x) f(x).
\]
Taking $y(x)=x f(x)$ and dividing by $\sigma^2/2$, leads to \eqref{e:chek}. Finally, observe that
\[
\frac{\sigma^2}{2}(e^{-x/b}y(x))'+\gamma e^{-x/b}y(x)=e^{-x/b}\frac{\sigma^2}{2}\left(y'(x)-\theta y(x)\right),
\]
whence the formula for $n_0(u)$ is also valid.
\end{proof}

\subsection{Constant burst rate $\varphi$}\label{constant}
If $\varphi$ is constant on $(0,\infty)$ and equal to $\kappa\ge 0$ then
the general solution of \eqref{e:chek} with $\theta> 0$
is $y(x)=w(\theta x)$ with $w$ of the form (see Remark~\ref{r:kummer})
\[
w(z)=c_1z {}_1F_1(\mu+1,2,z)+c_2 U(\mu,0,z),\quad \mu=\frac{2\kappa}{\sigma^2\theta},
\]
where ${}_1F_1$, $U$ are the confluent hypergeometric functions of the first and the second type, respectively,
and $c_1,c_2$ are constants. Since both Kummer functions with given parameters are bounded near $0$, the integrable $u$ in \eqref{e:invd} has to be of the form
\begin{equation}
 u(x) = c_0 e^{-x/b} {}_1F_1(\mu+1,2,\theta x), \quad x>0,\label{s:constant}
\end{equation}
where $c_0$ is a  constant.
Note that if $\theta<0$ then, by the Kummer transformation
    \[
    {}_1F_1(a_1,b_1,z)=e^z {}_1F_1(b_1-a_1,b_1,-z),
    \]
 we get
\[
u(x)=c_0 e^{-2\gamma x/\sigma^2} {}_1F_1(1-\mu,2,-\theta x).
\]

Since we have $n_0(u)=c_0\sigma^2/2$ and
$n_0(u)=\varphi(0)v$, we obtain that 
\[
v=\frac{\sigma^2}{2\varphi(0)}c_0,
\]
where the constant $c_0$ should be chosen such that
\[
\int_{[0,\infty)}g(x)m(dx)=\int_0^\infty u(x)dx +v=1.
\]
Consequently, we take $c_0$ satisfying
\[
\frac{\sigma^2}{2\varphi(0)}c_0+c_0c=1,
\]
where the normalization constant $c$ can be determined analytically, and is
\[
c=\int_0^{\infty} e^{-x/b} {}_1F_1(\mu+1,2,\theta x) dx=\left\{
                                                \begin{array}{ll}
                                                \frac{1}{b} {}_2F_1(\mu+1,1;2;b\theta) , & \hbox{ if }\theta>0 \\
                                                \frac{2\gamma}{\sigma^2} {}_2F_1(1-\mu,1;2;-\frac{2\gamma}{\sigma^2}\theta)  , & \hbox{ if }\theta<0,
                                                \end{array}
                                              \right. \]
with ${}_2F_1$ being  Gauss' hypergeometric function
\[
{}_2F_1(a_1,a_2;b_1;z)=\sum_{n=0}^\infty \frac{(a_1)_n (a_2)_n}{(b_1)_n}\frac{z^n}{n!}.
\]

When $\theta=0$ then equation \eqref{e:chek} has a solution of the form
\[
y(x)=\frac{2\kappa}{\sigma^2}x \sum_{n=0}^\infty\frac{(\frac{2\kappa}{\sigma^2}x)^n}{(n+1)!n!}.
\]
Thus
\[
u(x) = c_0  e^{-x/b} {}_0F_1(2, \frac{2\kappa}{\sigma^2} x),
\]
where  ${}_0F_1$ is the confluent hypergeometric limit function
\[
{}_0F_1(b_1,z)=\sum_{k=0}^\infty \frac{z^n}{(b_1)_n n!}=\lim_{a_1\to\infty}{}_1F_1(a_1,b_1,z/a_1).
\]
The constant $c$ is now
\[
c=\int_0^{\infty} e^{-x/b} {}_0F_1(2, \frac{2\kappa}{\sigma^2}x) dx.
\]

\section{Discussion and future directions}\label{sec:disc-conclusions}

Here we have treated, in great generality, the combined effects of both transcriptional/translational bursting as well as diffusive fluctuations on the dynamics of simple gene regulatory networks.  We have proved, under very mild conditions, the existence and uniqueness of a stationary density of molecular concentration as well as its asymptotic stability, and were able to provide an explicit expression for this density when the burst rate $\varphi$ is constant.

In \cite{mackeysimple} it has been argued, based on molecular interactions, that in general repressible and inducible systems the function $\varphi$ should have the form, c.f. also \cite{mcmmtkry11},
\begin{equation}
\varphi(x)=\lambda\frac{1+\Theta x^N}{\Lambda+\Omega x^N} \equiv \lambda f(x),
\label{eq:gen-control}
\end{equation}
where $\lambda,\Lambda,\Omega, N$ are positive constants and $\Theta\ge 0$.  All of these constants are determined by the reaction rate constants for molecular binding and unbinding, and ideally we would like to be able to offer an explicit solution in this case but we have been unable to do so.  We have explored various avenues, and conclude that this must remain an avenue for further research.

\section*{Acknowledgments}
MCM is supported by a Discovery Grant from the Natural Sciences and Engineering Research Council (NSERC) of Canada.  MCM would like to thank the
Institut f\"{u}r Theoretische Neurophysik, Universit\"{a}t Bremen for their hospitality during the time in which much of the writing of this paper took place.
MTK is supported by the Polish NCN grant No. 2017/27/B/ST1/00100.
This work was partially supported by the grant 346300 for IMPAN from the Simons Foundation and the matching 2015-2019 Polish MNiSW fund.



\end{document}